\documentclass{article}[11pt]

\usepackage{a4wide}
\usepackage[top=3cm]{geometry}

\usepackage{amsmath,amssymb}
\usepackage{amsthm}
\usepackage{tabls}
\usepackage{graphicx}
\usepackage{wrapfig}
\usepackage{array}
\usepackage{url}
\usepackage{multirow}
\usepackage{multicol}
\usepackage{here}

\newtheorem{theorem}{Theorem}
\newtheorem{lemma}{Lemma}
\newtheorem{observation}{Observation}
\newtheorem{corollary}{Corollary}

\theoremstyle{definition}
\newtheorem{problem}{Problem}

\theoremstyle{remark}
\newtheorem{remark}{Remark}

\newcommand{\rle}[1]{\mathit{RLE}_{#1}}
\newcommand{\rlebp}[1]{\mathit{Rb}[#1]}
\newcommand{\rleep}[1]{\mathit{Re}[#1]}
\newcommand{\rletree}{\mathsf{e^2rtre^2}}
\newcommand{\rleSigma}{\Sigma_{\mathit{RLE}}}
\newcommand{\rlesigma}{\sigma_{\mathit{RLE}}}
\newcommand{\eertree}{\mathsf{eertree}}

\newcommand{\emax}{\mathsf{emax}}
\newcommand{\esecond}{\mathsf{esec}}

\newcommand{\nil}{\mathsf{nil}}

\newcommand{\MaxPal}{\mathbf{MaxPal}}

\newcommand{\RBSP}{\mathbf{RBSP}}

\newcommand{\CPal}{\mathit{CPal}}

\newcommand{\sups}{\mathit{SUPS}}
\newcommand{\supss}{\mathit{SUPS}\mbox{s}}

\newcommand{\mupss}{\mathit{MUPS}\mbox{s}}
\newcommand{\M}{\mathcal{M}}

\newcommand{\mupsbeg}{\mathit{M_{beg}}}
\newcommand{\mupsend}{\mathit{M_{end}}}
\newcommand{\mupslen}{\mathit{M_{len}}}
\newcommand{\rmq}[3]{\mathsf{RmQ}_{#1}(#2, #3)}
\newcommand{\Pred}[2]{\mathsf{Pred}_{#1}(#2)}
\newcommand{\Succ}[2]{\mathsf{Succ}_{#1}(#2)}
\newcommand{\num}{\mathit{num}}

\newcommand{\mrb}{\mathit{mrb}}
\newcommand{\mre}{\mathit{mre}}

\newcommand{\rev}[1]{{#1}^{\mathit{R}}}
\newcommand{\occ}{\mathit{occ}}

\title{
  Fast Algorithms for the Shortest Unique Palindromic Substring Problem on Run-Length Encoded Strings
}

\author{
  {Kiichi Watanabe} \\
  Department of Informatics, Kyushu University, Japan \\
  \texttt{kiichi.watanabe@inf.kyushu-u.ac.jp}  
  \and
  {Yuto Nakashima} \\
  Department of Informatics, Kyushu University, Japan \\
  \texttt{yuto.nakashima@inf.kyushu-u.ac.jp}    
  \and
  {Shunsuke Inenaga} \\
  Department of Informatics, Kyushu University, Japan \\
  PRESTO, Japan Science and Technology Agency, Japan \\ 
  \texttt{inenaga@inf.kyushu-u.ac.jp}      
  \and 
  {Hideo Bannai} \\
  Department of Informatics, Kyushu University, Japan \\
  \texttt{bannai@inf.kyushu-u.ac.jp}
  \and 
  {Masayuki Takeda} \\
  Department of Informatics, Kyushu University, Japan \\
  \texttt{takeda@inf.kyushu-u.ac.jp}
}

\begin{document}

\maketitle

\begin{abstract}
	For a string $S$, 
	a palindromic substring $S[i..j]$ is said to be a \emph{shortest unique palindromic substring} ($\mathit{SUPS}$) 
	for an interval $[s, t]$ in $S$,
	if $S[i..j]$ occurs exactly once in $S$, the interval $[i, j]$ contains $[s, t]$, 
	and	every palindromic substring containing $[s, t]$ which is shorter than $S[i..j]$ 
	occurs at least twice in $S$.
  In this paper, we study the problem of answering $\mathit{SUPS}$ queries on run-length encoded strings.
  We show how to preprocess a given run-length encoded string $\rle{S}$ of size $m$ 
  in $O(m)$ space and $O(m \log \sigma_{\rle{S}} + m \sqrt{\log m / \log\log m})$ time 
  so that all $\mathit{SUPSs}$ for any subsequent query interval can be answered 
  in $O(\sqrt{\log m / \log\log m} + \alpha)$ time, where $\alpha$ is the number of outputs, 
  and $\sigma_{\rle{S}}$ is the number of distinct runs of $\rle{S}$.
  Additionaly, we consider a variant of the SUPS problem where
  a query interval is also given in a run-length encoded form.
  For this variant of the problem,
  we present two alternative algorithms with faster queries.
  The first one answers queries in $O(\sqrt{\log\log m /\log\log\log m} + \alpha)$ time
  and can be built in $O(m \log \sigma_{\rle{S}} + m \sqrt{\log m / \log\log m})$ time,
  and the second one answers queries in $O(\log \log m + \alpha)$ time
  and can be built in $O(m \log \sigma_{\rle{S}})$ time.
  Both of these data structures require $O(m)$ space.
\end{abstract}

%%% Input Sources %%%
\section{Introduction}

The \emph{shortest unique substring} (\emph{SUS}) problem,
which is formalized below, is a recent trend in the string processing community.
Consider a string $S$ of length $n$.
A substring $X = S[i..j]$ of $S$ is called a SUS for a position $p$~($1 \leq p \leq n$)
iff the interval $[i..j]$ contains $p$, $X$ occurs in $S$ exactly once, 
and	every substring containing $p$ which is shorter than $S[i..j]$ occurs at least twice in $S$.
The SUS problem is to preprocess a given string $S$ so that
SUSs for query positions $p$ can be answered quickly.
The study on the SUS problem was initiated by Pei et al.,
and is motivated by an application to bioinformatics e.g.,
designing polymerase chain reaction (PCR) primer~\cite{Pei}.
Pei et al.~\cite{Pei} showed an $\Theta(n^2)$-time and space
preprocessing scheme such that all $k$ SUSs for
a query position can be answered in $O(k)$ time.
Later, two independent groups, Tsuruta et al.~\cite{Tsuruta} and Ileri et al.~\cite{ileri14_SUS},
showed algorithms that use $\Theta(n)$ time and space\footnote{Throughout this paper, we measure the space complexity of an algorithm with the number of \emph{words} that the algorithm occupies in the word RAM model, unless otherwise stated.} for preprocessing,
and all SUSs can be answered in $O(k)$ time per query.
To be able to handle huge text data where $n$ can be massively large,
there have been further efforts to reduce the space usage.
Hon et al.~\cite{HonTX15} proposed an ``in-place'' algorithm which works within
space of the input string $S$ and two output arrays $A$ and $B$ of length $n$ each,
namely, in $n \log_2 \sigma$ \emph{bits} plus $2n$ words of space.
After the execution of their algorithm that takes $O(n)$ time,
the beginning and ending positions of a SUS for each text position
$i$~($1 \leq i \leq n$) are consequently stored in $A[i]$ and $B[i]$, respectively,
and $S$ remains unchanged.
Hon et al.'s algorithm can be extended to handle
SUSs with approximate matches, with a penalty of $O(n^2)$ preprocessing time.
For a pre-determined parameter $\tau$,
Ganguly et al.~\cite{GangulyHST17_SUS}
proposed a time-space trade-off algorithm for the SUS problem
that uses $O(n / \tau)$ additional working space (apart from the input string $S$)
and answers each query in $O(n \tau^2 \log \frac{n}{\tau})$ time.
They also proposed a ``succinct'' data structure of
$4n + o(n)$ \emph{bits} of space
that can be built in $O(n \log n)$ time and can answer a SUS for each given query
position in $O(1)$ time.
Another approach to reduce the space requirement for the SUS problem
is to work on a ``compressed'' representation of the string $S$.
Mieno et al.~\cite{MienoIBT16} developed a data structure of $\Theta(m)$ space
(or $\Theta(m \log n)$ \emph{bits} of space)
that answers all $k$ SUSs for a given position in $O(\sqrt{\log m / \log \log m} + k)$ time,
where $m$ is the size of the \emph{run length encoding} (\emph{RLE})
of the input string $S$.
This data structure can be constructed in $O(m \log m)$ time with $O(m)$ words of working space
if the input string $S$ is already compressed by RLE,
or in $O(n + m \log m)$ time with $O(m)$ working space
if the input string $S$ is given without being compressed.

A generalized version of the SUS problem,
called the \emph{interval} SUS problem, is to answer
SUSs that contain a query interval $[s, t]$ with $1 \leq s \leq t \leq n$.
Hu et al.~\cite{HuPT14} proposed an optimal $\Theta(n)$ time and space algorithm
to preprocess a given string $S$
so that all $k$ SUSs for a given query interval are reported in $O(k)$ time.
Mieno et al.'s data structure~\cite{MienoIBT16}
also can answer interval SUS queries
with the same preprocessing time/space and query time as above.

Recently, a new variant of the SUS problem, called the
\emph{shortest unique palindromic substring} (\emph{SUPS}) problem is considered~\cite{SUPS}.
A substring $P = S[i..j]$ is called a SUPS for an interval $[s,t]$
iff $P$ occurs exactly once in $S$, $[s,t] \subseteq [i,j]$,
and every palindromic substring of $S$ which contains interval $[s,t]$ and is shorter
than $P$ occurs at least twice in $S$.
The study on the SUPS problem is motivated by an application in molecular biology.
Inoue et al.~\cite{SUPS} showed how to preprocess a given string $S$ of length $n$
in $\Theta(n)$ time and space so that all $\alpha$ SUPSs (if any)
for a given interval can be answered 
in $O(\alpha + 1)$ time\footnote{It is possible that $\alpha = 0$ for some intervals.}.
While this solution is optimal in terms of the length $n$ of the input string,
no space-economical solutions for the SUPS problem were known.

In this paper, we present the \emph{first} space-economical solution
to the SUPS problem based on RLE.
The proposed algorithm computes
a data structure of $\Theta(m)$ space
that answers each SUPS query in $O(\sqrt{\log m / \log \log m}+\alpha)$ time.
The most interesting part of our algorithm is how to preprocess
a given RLE string of length $m$ in $O(m (\log \sigma_{RLE_S} + \sqrt{\log m / \log \log m}))$ time,
where $\sigma_{RLE_S}$ is the number of distinct runs in the RLE of $S$.
Note that $\sigma_{RLE} \leq m$ always holds.
For this sake, we propose RLE versions of 
Manacher's maximal palindrome algorithm~\cite{Manacher75} and Rubinchik and Shur's eertree data structure~\cite{EERTREE},
which may be of independent interest.
We remark that our preprocessing scheme is quite different from
Mieno et al.'s method~\cite{MienoIBT16} for the SUS problem on RLE strings
and Inoue et al.'s method~\cite{SUPS} for the SUPS problem on plain strings.

  Additionaly, we consider a variant of the SUPS problem where
  a query interval is also given in a run-length encoded form.
  For this variant of the problem,
  we present two alternative algorithms with faster queries.
  The first one answers queries in $O(\sqrt{\log\log m /\log\log\log m} + \alpha)$ time
  and can be built $O(m \log \sigma_{\rle{S}} + m \sqrt{\log m / \log\log m})$ time,
  and the second one answers queries in $O(\log \log m + \alpha)$ time
  and can be built in $O(m \log \sigma_{\rle{S}})$ time.
  Both of these data structures require $O(m)$ space.

A part of the results presented here appeared in a preliminary version of this paper~\cite{WatanabeNIBT19}.

\section{Preliminaries}

\subsection{Strings}
Let $\Sigma$ be an ordered {\em alphabet} of size $\sigma$.
An element of $\Sigma^*$ is called a {\em string}.
The length of a string $S$ is denoted by $|S|$.
The empty string $\varepsilon$ is a string of length 0.
%Let $\Sigma^+$ be the set of non-empty strings,
%i.e., $\Sigma^+ = \Sigma^* - \{\varepsilon \}$.
For a string $S = XYZ$, $X$, $Y$ and $Z$ are called
a \emph{prefix}, \emph{substring}, and \emph{suffix} of $S$, respectively.
% A prefix $x$ of $w$ is called a \emph{proper prefix} of $w$ if $x \neq w$.
The $i$-th character of a string $S$ is denoted by $S[i]$, for $1 \leq i \leq |S|$.
Let $S[i..j]$ denote the substring of $S$ that begins at position $i$ and ends at
position $j$, for $1 \leq i \leq j \leq |S|$.
For convenience, let $S[i..j] = \varepsilon$ for $i > j$.

For any string $S$, let $\rev{S} = S[|S|] \cdots S[1]$
denote the reversed string of $S$.
A string $P$ is called a \emph{palindrome} iff $P = \rev{P}$.
A substring $P = S[i..j]$ of a string $S$ is called a \emph{palindromic substring}
iff $P$ is a palindrome.
For a palindromic substring $P = S[i..j]$,
$\frac{i+j}{2}$ is called the \emph{center} of $P$.
A palindromic substring $P = S[i..j]$ is said to be a \emph{maximal palindrome} of $S$,
iff $S[i-1] \neq S[j+1]$, $i = 1$ or $j = |S|$.
A suffix of string $S$ that is a palindrome is called a \emph{suffix palindrome} of $S$.
Clearly any suffix palindrome of $S$ is a maximal palindrome of $S$.

We will use the following lemma in the analysis of our algorithm.
\begin{lemma}[\cite{DroubayJP01}] \label{lem:distinct_pal}
Any string of length $k$ can contain at most $k+1$
distinct palindromic substrings (including the empty string $\varepsilon$).
\end{lemma}

\subsection{MUPSs and SUPSs}

For any strings $X$ and $S$,
let $\occ_S(X)$ denote the number of occurrences of $X$ in $S$,
i.e., $\occ_S(X) = |\{i \mid S[i..i+|X|-1] = X\}|$.
A string $X$ is called a \emph{unique} substring of a string $S$
iff $\occ_S(X) = 1$.
A substring $P = S[i..j]$ of string $S$
is called a \emph{minimal unique palindromic substring} (\emph{MUPS})
of a string $S$ iff
(1) $P$ is a unique palindromic substring of $S$ and
(2) either $|P| \geq 3$ and the palindrome $Q = S[i+1..j-1]$ satisfies
$\occ_S(Q) \geq 2$, or $1 \leq |P| \leq 2$.

\begin{lemma}[\cite{SUPS}] \label{lem:MUPS_do_not_nest}
  MUPSs do not nest, namely, 
  for any pair of distinct MUPSs,
  one cannot contain the other.
\end{lemma}

Due to Lemma~\ref{lem:MUPS_do_not_nest},
both of the beginning positions and the ending positions of MUPSs
are monotonically increasing.
Let $\M_S$ denote the list of MUPSs in $S$
sorted in increasing order of their beginning positions
(or equivalently the ending positions) in $S$.

Let $[s,t]$ be an integer interval over the positions in a string $S$,
where $1 \leq s \leq t \leq |S|$.
A substring $P = S[i..j]$ of string $S$ is called
a \emph{shortest unique palindromic substring} ($\emph{SUPS}$)
for interval $[s,t]$ of $S$,
iff
(1) $P$ is a unique palindromic substring of $S$,
(2) $[s, t] \subseteq [i, j]$, and
(3) there is no unique palindromic substring $Q = S[i'..j']$
such that $[s, t] \subseteq [i', j']$ and $|Q| < |P|$.
We give an example in Fig.~\ref{fig:sups-mups}.
\begin{figure}[t]
  \centerline{\includegraphics[width=0.9\textwidth]{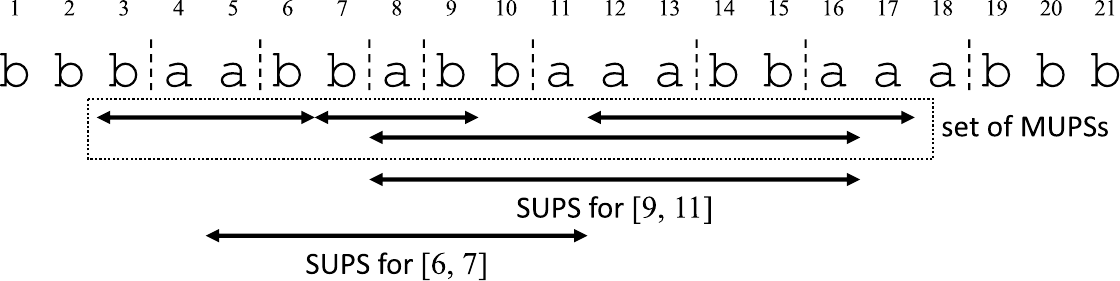}}
  \caption{
  	This figure shows all $\mupss$ and some $\supss$ of 
  	string $\rle{S} = \mathtt{b}^3\mathtt{a}^2\mathtt{b}^2\mathtt{a}^1\mathtt{b}^2\mathtt{a}^3\mathtt{b}^2\mathtt{a}^3\mathtt{b}^3$.
  	There are 4 $\mupss$ illustrated in the box.
  	The $\sups$ for interval $[6, 7]$ is $S[5..11]$, 
  	and the $\sups$ for interval $[9, 11]$ is $S[8..16]$.
  }
  \label{fig:sups-mups}
\end{figure}

\subsection{Run length encoding (RLE)}

The {\em run-length encoding} $\rle{S}$ of string $S$ is
a compact representation of $S$ such that
each maximal run of the same characters in $S$ is represented
by a pair of the character and the length of the run.
More formally, let $\mathcal{N}$ denote the set of positive integers.
For any non-empty string $S$, 
$\rle{S} = (a_1, e_1),  \ldots, (a_m, e_m)$,
where $a_j \in \Sigma$ and $e_j \in \mathcal{N}$ for any $1 \leq j \leq m$,
and $a_j \neq a_{j+1}$ for any $1 \leq j < m$.
E.g., if $S = \mathtt{aacccccccbbabbbb}$,
then $\rle{S} = (\mathtt{a}, 2), (\mathtt{c}, 7), (\mathtt{b}, 2), (\mathtt{a}, 1), (\mathtt{b}, 4)$.
Each $(a, e)$ in $\rle{S}$ is called a (character) \emph{run},
and $e$ is called the exponent of this run.
We also denote each run by $a^{e}$
when it seems more convenient and intuitive.
For example, we would write as $(a, e)$ when it seems more convenient
to treat it as a kind of character (called an RLE-character),
and would write as $a^e$ when it seems more convenient
to treat it as a string consisting of $e$ $a$'s.

The \emph{size} of $\rle{S}$ is the number $m$ of runs in $\rle{S}$.
Let $\rlebp{j}$ (resp. $\rleep{j}$) denote the beginning (resp. ending) position 
of the $j$th run in the string $S$,
i.e., $\rlebp{j} = 1+\sum_{i=0}^{j-1} e_i$ with $e_0 = 0$
and $\rleep{j} = \sum_{i=1}^{j}e_i$.
The \emph{center} of the $j$th run is $\frac{\rlebp{j}+\rleep{j}}{2}$.

For any two ordered pairs $(a, e), (a', e') \in \Sigma \times \mathcal{N}$
of a character and positive integer,
we define the equality such that $(a, e) = (a', e')$
iff $a = a'$ and $e = e'$ both hold.
We also define a total order of these pairs
such that $(a, e) < (a', e')$ iff $a < a'$,
or $a = a'$ and $e < e'$.

An occurrence of a palindromic substring
$P = S[i..i']$ of a string $S$ with $\rle{S}$ of size $m$ is
said to be \emph{RLE-bounded} 
if $i = \rlebp{j}$ and $i' = \rleep{j'}$ for some $1 \leq j \leq j' \leq m$,
namely, if both ends of the occurrence
touch the boundaries of runs.
An RLE-bounded occurrence 
$P = S[i..i']$ is said to be \emph{RLE-maximal} 
if $(a_{j-1}, e_{j-1}) \neq (a_{j'+1}, e_{j'+1})$,
$j = 1$ or $j' = m$.
Note that an RLE-maximal occurrence
of a palindrome may not be maximal in the string $S$.
E.g., consider string $S = \mathtt{caabbcccbbaaaac}$ with $\rle{S} = \mathtt{c^1a^2b^2c^3b^2a^4c^1}$.
\begin{itemize}
\item The occurrence of palindrome $\mathtt{c^3}$ is RLE-bounded
         but is neither RLE-maximal nor maximal.
\item The occurrence of palindrome $\mathtt{b^2c^3b^2}$ is RLE-maximal but is not maximal.
\item The occurrence of palindrome $\mathtt{a^2b^2c^3b^2a^2}$ is not RLE-maximal but is maximal.
\item The first (leftmost) occurrence of palindrome $\mathtt{a^2}$ is both RLE-maximal and maximal.
\end{itemize}

\subsection{Problem}

In what follows, we assume that our input strings are given as RLE strings.
In this paper, we tackle the following problem.

\begin{problem}[$\sups$ problem on run-length encoded strings]
\label{prob:SUPS_RLE}
\leavevmode %\par
\begin{description}
\item[Preprocess:] $\rle{S} = (a_1, e_1), \ldots, (a_m, e_m)$ of size $m$ representing a string $S$ of length $n$.
\item[Query:] An integer interval $[s, t]$~$(1 \leq s \leq t \leq n)$.
\item[Return:] All SUPSs for interval $[s, t]$.
\end{description}
\end{problem}

In case the string $S$ is given as a plain string of length $n$,
then the time complexity of our algorithm will be increased by 
an additive factor of $n$ that is needed to compute $\rle{S}$,
while the space usage will stay the same since $\rle{S}$ can be computed in
constant space.

\section{Computing MUPSs from RLE strings}
\label{sec:MUPS}

The following known lemma suggests that
it is helpful to compute the set $\mathcal{M}_S$ of MUPSs of $S$
as a preprocessing for the SUPS problem.

\begin{lemma}[\cite{SUPS}] \label{lem:unimups_in_sups}
  For any SUPS $S[i..j]$ for some interval,
  there exists exactly one MUPS that is contained in the interval $[i,j]$.
  Furthermore, the MUPS has the same center as the SUPS $S[i..j]$.
\end{lemma}

\subsection{Properties of MUPSs on RLE strings}

Now we present some useful properties of MUPSs
on the run-length encoded string $\rle{S} = (a_1, e_1), \ldots, (a_m, e_m)$.

\begin{lemma} \label{lem:mups_center}
For any MUPS $S[i..j]$ in $S$, 
there exists a unique integer $k$~($1 \leq k \leq m$)
such that $\frac{i+j}{2} = \frac{\rlebp{k}+\rleep{k}}{2}$.
\end{lemma}

\begin{proof}
	Suppose on the contrary that there is a MUPS $S[i..j]$ 
	such that $\frac{i+j}{2} \neq \frac{\rlebp{k}+\rleep{k}}{2}$ for any $1 \leq k \leq m$.
	Let $l$ be the integer that satisfies $\rlebp{l} \leq \frac{i+j}{2} \leq \rleep{l}$.
	By the assumption, the longest palindrome 
	whose center is $\frac{i+j}{2}$ is $a_l^{\min\{i-\rlebp{l}+1,\rleep{l}-j+1\}}$.
           However, this palindrome $a_l^{\min\{i-\rlebp{l}+1,\rleep{l}-j+1\}}$ occurs
           at least twice in the $l$th run $a_l^{e_l}$.
	Hence MUPS $S[i..j]$ is not a unique palindromic substring, a contradiction.
\end{proof}

The following corollary is immediate from Lemma~\ref{lem:mups_center}.
\begin{corollary} \label{coro:num-of-mupss}
	For any string $S$, $|\M_S| \leq m$.
\end{corollary}
It is easy to see that the above bound is tight:
for instance, any string where each run has a distinct character
(i.e., $m = \rlesigma$) contains exactly $m$ MUPSs.
Our preprocessing and query algorithms which will follow
are heavily dependent on this lemma and corollary.

\subsection{RLE version of Manacher's algorithm} \label{subsec:manacher}

Due to Corollary~\ref{coro:num-of-mupss},
we can restrict ourselves to computing palindromic substrings
whose center coincides with the center of each run.
These palindromic substrings are called \emph{run-centered} palindromes.
Run-centered palindromes will be candidates of MUPSs of the string $S$.

To compute run-centered palindromes from $\rle{S}$,
we utilize Manacher's algorithm~\cite{Manacher75}
that computes
all maximal palindromes for a given (plain) string of length $n$
in $O(n)$ time and space.
Manacher's algorithm is based only on character equality comparisons,
and hence it works with general alphabets.

Let us briefly recall how Manacher's algorithm works.
It processes a given string $S$ of length $n$ from left to right.
It computes an array $\MaxPal$ of length $2n-1$ such that
$\MaxPal{}[c]$ stores
the length of the maximal palindrome with center $c$
for $c = 1, 1.5, 2, \ldots, n-1, n-0.5, n$.
Namely, Manacher's algorithm processes a given string $S$ in an online manner
from left to right.
This algorithm is also able to compute, for each position $i = 1, \ldots, n$,
the longest palindromic suffix of $S[1..i]$ in an online manner.

Now we apply Manacher's algorithm to our run-length encoded input string
$\rle{S} = (a_1, e_1), \ldots,$ $(a_m, e_m)$.
Then, what we obtain after the execution of Manacher's algorithm
over $\rle{S}$ is all \emph{RLE-maximal} palindromes of $S$.
Note that by definition all RLE-maximal palindromes are run-centered.
Since $\rle{S}$ can be regarded as a string of length $m$
over an alphabet $\Sigma \times \mathcal{N}$,
this takes $O(m)$ time and space.

\begin{remark} \label{rmk:maximal_extension}
  If wanted, we can compute all \emph{maximal} palindromes
  of $S$ in $O(m)$ time after the execution of Manacher's algorithm to $\rle{S}$.
  First, we compute every run-centered
  maximal palindrome
  $P_l$ that has its center in each $l$th run in $\rle{S}$.
  For each already computed run-centered RLE-maximal palindrome
  $Q_l = S[\rlebp{i}..\rleep{j}]$ with $1 < i \leq j < m$,
  it is enough to first check whether $a_{i-1} = a_{j+1}$.
  If no, then $P_l = Q_l$,
  and if yes then we can further extend both
  ends of $Q_l$ with $(a_{i-1}, \min\{e_{i-1}, e_{j+1}\})$
  and obtain $P_l$.
  As a side remark,
  we note that any other maximal palindromes of $S$ are not run-centered,
  which means that any of them consists only of the same characters
  and lie inside of one character run.
  Such maximal palindromes are trivial and need not be explicitly computed.
\end{remark}

\subsection{RLE version of eertree data structure}

The \emph{eertree}~\cite{EERTREE} of a string $S$,
denoted $\eertree(S)$,
is a pair of two rooted trees $\mathsf{T}_{\mathrm{odd}}$ and $\mathsf{T}_\mathrm{even}$
which represent all distinct palindromic substrings of $S$.
The root of $\mathsf{T}_{\mathrm{odd}}$ represents the empty string $\varepsilon$
and each non-root node of $\mathsf{T}_{\mathrm{odd}}$ 
represents a non-empty palindromic substring of $S$ of odd length.
Similarly, the root of $\mathsf{T}_\mathrm{even}$ represents
the empty string $\varepsilon$ and each non-root node of 
$\mathsf{T}_\mathrm{even}$ represents a non-empty palindromic substring of $S$ of even length.
From the root $r$ of $\mathsf{T}_{\mathrm{odd}}$,
there is a labeled directed edge $(r, a, v)$
if $v$ represents a single character $a \in \Sigma$.
For any non-root node $u$ of $\mathsf{T}_{\mathrm{odd}}$ or $\mathsf{T}_\mathrm{even}$,
there is a labeled directed edge $(u, a, v)$ from $u$ to node $v$
with character label $a \in \Sigma$ if $aua = v$.
For any node $u$,
the labels of out-going edges of $u$ must be mutually distinct.

By Lemma~\ref{lem:distinct_pal},
any string $S$ of length $n$ can contain at most $n+1$
distinct palindromic substrings (including the empty string $\varepsilon$).
Thus, the size of $\eertree(S)$ is linear in the string length $n$.
Rubinchik and Shur~\cite{EERTREE} showed how to construct
$\eertree(S)$ in $O(n \log \sigma_S)$ time and $O(n)$ space,
where $\sigma_S$ is the number of distinct characters in $S$.
They also showed how to compute the number of occurrences of
each palindromic substring in $O(n \log \sigma_S)$ time and $O(n)$ space,
using $\eertree(S)$.

Now we introduce a new data structure named \emph{RLE-eertrees} based on eertrees.
Let $\rle{S}  = (a_1, e_1), \ldots, (a_m, e_m)$,
and let $\rleSigma$ be the set of maximal runs of $S$,
namely, $\rleSigma = \{(a, e) \mid (a, e) = (a_i, e_i) \mbox{ for some } 1 \leq i \leq m\}$.
Let $\sigma_{RLE} = |\rleSigma|$.
Note that $\sigma_{RLE} \leq m$ always holds.
The RLE-eertree of string $S$, denoted by $\rletree(S)$,
is a \emph{single} eertree $\mathsf{T}_{\mathtt{odd}}$
\emph{over the RLE alphabet $\rlesigma \subset \Sigma \times \mathcal{N}$},
which represents distinct run-centered palindromes of $S$ which have
an RLE-bounded occurrence $[i,i']$ such that $i = \rlebp{j}$ and $i' = \rleep{j'}$
for some $1 \leq j \leq j' \leq m$
(namely, the both ends of the occurrence touch the boundary of runs),
or an occurrence as a maximal palindrome in $S$.
We remark that the number of \emph{runs} in each palindromes in $\rletree(S)$ is odd,
but their decompressed string length may be odd or even.
In $\rletree(S)$,
there is a directed labeled edge $(u, a^e, v)$ from node $u$ to node $v$
with label $a^e \in \Sigma_{RLE}$
if (1) $a^eua^e = v$,
or (2) $u = \varepsilon$ and $v = a^e \in \Sigma \times \mathcal{N}$.
Note that if the in-coming edge of a node $u$ is labeled with $a^e$,
then any out-going edge of $u$ cannot have a label $a^f$
with the same character $a$.
Since $\rletree{S}$ is an eertree over the alphabet $\Sigma_{RLE}$ of size $\rlesigma \leq m$,
it is clear that the number of out-going edges of each node is bounded by $\rlesigma$.
We give an example of $\rletree(S)$ in Fig.~\ref{fig:rletree}.

\begin{figure}[t]
  \centerline{\includegraphics[width=0.9\textwidth]{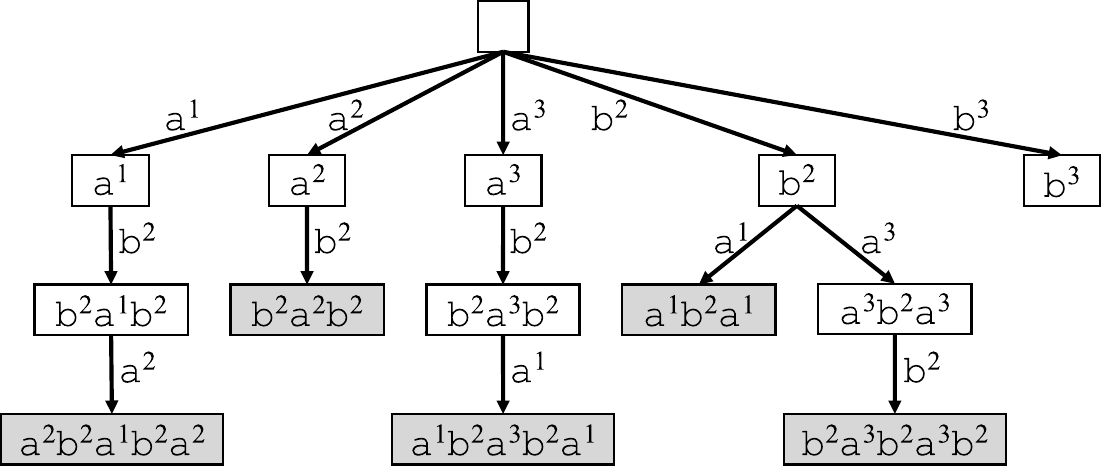}}
  \caption{
  	The RLE-eertree $\rletree(S)$ of $\rle{S} = 
  	\mathtt{b}^3\mathtt{a}^2\mathtt{b}^2\mathtt{a}^1\mathtt{b}^2\mathtt{a}^3\mathtt{b}^2\mathtt{a}^3\mathtt{b}^3$.
        Each white node represents a run-centered palindromic substring
        of $S$ that has an RLE-bounded occurrence,
        while each gray node represents a run-centered palindromic substring
        of $S$ that has a maximal occurrence in $S$. 
  }
  \label{fig:rletree}
\end{figure}

\begin{lemma}
  Let $S$ be any string of which the size of $\rle{S}$ is $m$.
  Then, the number of nodes in $\rletree(S)$ is at most $2m+1$.
\end{lemma}

\begin{proof}
  First, we consider $\rle{S}$ as a string of length $m$ over
  the alphabet $\Sigma_{RLE}$.
  It now follows from Lemma~\ref{lem:distinct_pal}
  that the number of non-empty distinct run-centered palindromic substrings of $S$
  that have an RLE-bounded occurrence is at most $m$.
  Each of these palindromic substrings are represented by
  a node of $\rletree(S)$,
  and let $\rletree(S)'$ denote the tree
  consisting only of these nodes
  (in the example of Fig.~\ref{fig:rletree},
  $\rletree(S)'$ is the tree consisting only of the white nodes).

  Now we count the number of nodes in $\rletree(S)$
  that do not belong to $\rletree(S)'$ 
  (the gray nodes in the running example of Fig.~\ref{fig:rletree}).
  Since each palindrome represented by this type of node
  has a run-centered maximal occurrence in $S$,
  the number of such palindromes is bounded by the number $m$ of runs in $\rle{S}$.
  
  Hence, including the root that represent the empty string,
  there are at most $2m+1$ nodes in $\rletree(S)$.
\end{proof}

\begin{lemma} \label{lem:rletree_constrction}
  Given $\rle{S}$ of size $m$,
  $\rletree(S)$ can be built in $O(m \log \rlesigma)$ time and $O(m)$ space,
  where the out-going edges of each node are sorted according to
  the total order of their labels.
  Also, in the resulting $\rletree(S)$,
  each non-root node $u$ stores the number of
  occurrences of $u$ in $S$ which are RLE-bounded or maximal.
\end{lemma}

\begin{proof}
  Our construction algorithm comprises three steps.
  We firstly construct $\rletree(S)'$,
  secondly compute an auxiliary array $\CPal$ that will be used for the next step,
  and thirdly we add some nodes
  that represent run-centered maximal palindromes which are not in $\rletree(S)'$
  so that the resulting tree forms the final structure $\rletree(S)$.

  Rubinchik and Shur~\cite{EERTREE} proposed an online algorithm
  which constructs $\eertree(T)$ of a string of length $k$ in
  $O(k \log \sigma_T)$ time with $O(k)$ space,
  where $\sigma_T$ denotes the number of distinct characters in $T$.
  They also showed how to store, in each node,
  the number of occurrences of the corresponding palindromic substring in $T$.
  Thus, the Rubinchik and Shur algorithm applied to
  $\rle{S}$ computes $\rletree(S)'$ in $O(m \log \rlesigma)$ time with $O(m)$ space.
  Also, now each node $u$ of $\rletree(S)'$ stores the number of
  \emph{RLE-bounded} occurrence of $u$ in $S$.
  This is the first step.

  The second step is as follows: 
  Let $\CPal$ be an array of length $m$ such that,
  for each $1 \leq i \leq m$,
  $\CPal[i]$ stores a pointer to the node in $\rletree(S)'$ 
  that represents the RLE-bounded palindrome centered at $i$.
  A simple application of the Rubinchik and Shur algorithm to
  $\rle{S}$ only gives us the leftmost occurrence of each RLE-bounded palindrome
  in $\rletree(S)'$.
  Hence, we only know the values of $\CPal$ in the positions
  that are the centers of the leftmost occurrences of RLE-bounded palindromes.
  To compute the other values in $\CPal$,
  we run Manacher's algorithm to $\rle{S}$ as in Section~\ref{subsec:manacher}.
  Since Manacher's algorithm processes $\rle{S}$ in an online manner from left to right,
  and since we already know the leftmost occurrences of all RLE-bounded palindromes
  in $\rle{S}$, we can copy the pointers from previous occurrences. 
  In case where an RLE-bounded palindrome extends with a newly read RLE-character $(a, e)$
  after it is copied from a previous occurrence during the execution of Manacher's algorithm,
  then we traverse the edge labeled $(a, e)$ from the current node of $\rletree(S)'$.
  By repeating this until the mismatch is found
  in extension of the current RLE-bounded palindrome,
  we can find the corresponding node for this RLE-bounded palindrome.
  This way we can compute $\CPal[i]$ for all $1 \leq i \leq m$ in $O(m \log \rlesigma)$ total
  time with $O(m)$ total space.

  In the third step, we add new nodes that represent
  run-centered maximal (but not RLE-bounded) palindromic substrings.
  For this sake, we again apply Manacher's algorithm to $\rle{S}$,
  but in this case it is done as in Remark~\ref{rmk:maximal_extension} of Section~\ref{subsec:manacher}.
  With the help of $\CPal$ array,
  we can associate each run $(a_l, e_l)$ with
  the RLE-bounded palindromic substring that has the center in $(a_l, e_l)$.
  Let $Q_l = S[\rlebp{i}..\rleep{j}]$ denote this palindromic substring for $(a_l, e_l)$,
  where $1 \leq i \leq l \leq j \leq m$,
  and $u_l$ the node that represents $Q_l$ in $\rletree(S)'$.
  We first check whether $a_{i-1} = a_{j+1}$.
  If no, then $Q_l$ does not extend from this run $(a_l, e_l)$,
  and if yes then we extend both ends of $Q_l$ with $(a_{i-1}, \min\{e_{i-1}, e_{j+1}\})$.
  Assume w.l.o.g. that $e_{i-1} = \min\{e_{i-1}, e_{j+1}\}$.
  If there is no out-going edge of $u_l$ with label $(a_{i-1}, e_{i-1})$,
  then we create a new child of $u_l$ with an edge labeled  $(a_{i-1}, e_{i-1})$.
  Otherwise, then let $v$ be the existing child of $u_l$
  that represents $a_{i-1}^{e_{i-1}} u_l a_{i-1}^{e_{i-1}}$.
  We increase the number of occurrences of $v$ by 1.
  This way, we can add all new nodes
  and we obtain $\rletree(S)$.
  Note that each node stores the number of RLE-bounded or maximal occurrence
  of the corresponding run-centered palindromic substring.
  It is easy to see that the second step takes a total of $O(m \log \rlesigma)$
  time and $O(m)$ space.
\end{proof}

It is clear that for any character $a \in \Sigma$,
there can be only one MUPS of form $a^e$.
Namely, $a^e$ is a MUPS iff $e$ is the largest exponent
for all runs of $a$'s in $S$ and $\occ_S(a^e) = 1$.
Below, we consider other forms of MUPSs.
Let $P$ be a non-empty palindromic substring of string $S$
that has a run-centered RLE-bounded occurrence.
For any character $a \in \Sigma$,
let $\emax$ and $\esecond$ denote the largest and second largest
positive integers such that
$a^{\emax}Pa^{\emax}$ and $a^{\esecond}Pa^{\esecond}$ are
palindromes that have 
run-centered RLE-bounded or maximal occurrences in $S$.
If such integers do not exist, then let $\emax = \nil$ and $\esecond = \nil$.
\begin{observation} \label{obs:MUPSonRLE}
  There is at most one MUPS of form $a^e P a^e$ in $S$. Namely,
  \begin{enumerate}
    \item[(1)] The palindrome $a^{\esecond+1}Pa^{\esecond+1}$ is a MUPS of $S$
                  iff $\emax \neq \nil$, $\esecond \neq \nil$, and $\occ_S(a^{\emax}Pa^{\emax})$ $= 1$.
    \item[(2)] The palindrome $a^{1}Pa^{1}$ is a MUPS of $S$
                  iff $\emax \neq \nil$, $\esecond = \nil$, and $\occ_S(a^{\emax}Pa^{\emax}) = 1$.
    \item[(3)] There is no MUPS of form $a^e P a^e$ with any $e \geq 1$
                  iff either $\emax = \nil$, or $\emax \neq \nil$ and $\occ_S(a^{\emax}Pa^{\emax}) > 1$.
  \end{enumerate}
\end{observation}

\begin{lemma}
$\M_S$ can be computed in $O(m \log \rlesigma)$ time and $O(m)$ space.
\end{lemma}

\begin{proof}
  For each node $u$ of $\rletree(S)$,
  let $\Sigma_u$ be the set of characters $a$
  such that there is an out-going edge of $u$ labeled by $(a, e)$
  with some positive integer $e$.
  Due to Observation~\ref{obs:MUPSonRLE},
  for each character in $\Sigma_u$,
  it is enough to check the out-going edges
  which have the largest and second largest exponents with character $a$.
  Since the edges are sorted,
  we can find all children of $u$ that represent MUPSs in
  time linear in the number of children of $u$.
  Hence, given $\rletree(S)$,
  it takes $O(m)$ total time to compute all MUPSs in $S$.
  $\rletree(S)$ can be computed in $O(m \log \rlesigma)$ time and $O(m)$ space
  by Lemma~\ref{lem:rletree_constrction}.

  What remains is how to sort the MUPSs in increasing order of their beginning positions.
  We associate each MUPS with the run where its center lies.
  Since each MUPS occurs in $S$ exactly once
  and MUPSs do not nest (Lemma~\ref{lem:MUPS_do_not_nest}),
  each run cannot contain the centers of two or more MUPSs.
  We compute an array $A$ of size $m$ such that
  $A[j]$ contains the corresponding interval of the MUPS
  whose center lies in the $j$th run in $\rle{S}$, if it exists.
  After computing $A$,
  we scan $A$ from left to right.
  Since again MUPSs do not nest,
  this gives as the sorted list $\M_S$ of MUPSs.
  It is clear that this takes a total of $O(m)$ time and space. 
\end{proof}

\section{SUPS queries on RLE strings}
\label{sec:query_algorithm}

In this section, we present our algorithm for SUPS queries.
Our algorithm is based on Inoue et al.'s algorithm~\cite{SUPS} 
for SUPS queries on a plain string.
The big difference is that
the space that we are allowed for is limited to $O(m)$.

\subsection{Data structures}\label{data_structure}
As was discussed in Section~\ref{sec:MUPS}, we can compute the list
$\M_S$ of all MUPSs of string $S$ efficiently.  We store $\M_S$ using
the three following arrays:
\begin{itemize}
	\item $\mupsbeg[i]$ : the beginning position of the $i$th MUPS in $\M_S$.
	\item $\mupsend[i]$ : the ending position the $i$th MUPS in $\M_S$.
	\item $\mupslen[i]$ : the length of the $i$th MUPS in $\M_S$.
\end{itemize}
Since the number of MUPSs in $\M_S$ is at most $m$
(Corollary~\ref{coro:num-of-mupss}),
the length of each array is at most $m$.
In our algorithm, we use \emph{range minimum queries} and
\emph{predecessor/successor queries} on integer arrays.

Let $A$ be an integer array of length $d$.
A range minimum query $\rmq{A}{i}{j}$ returns one of $\arg \min_{i \leq k \leq j}\{A[k]\}$
for a given interval $[i, j]$ in $A$.
\begin{lemma}[e.g.~\cite{rmqspace}]
	We can construct an $O(d)$-space data structure in $O(d)$ time for an integer array $A$ of length $d$
	which can answer $\rmq{A}{i}{j}$ in constant time for any query $[i, j]$.
\end{lemma}

Let $B$ be an array of $d$ positive integers in $[1, N]$ in increasing order.
The predecessor and successor queries on $B$ are defined for any $1 \leq k \leq N$ as follows.
\begin{eqnarray*}
  \Pred{B}{k} & = &
  \begin{cases}
    \max\{ i \mid B[i] \leq k\} \quad & \mbox{if it exists,}\\
    0 \quad & \mbox{otherwise.}
  \end{cases} \\
  \Succ{B}{k} & = &
  \begin{cases}
    \min\{ i \mid B[i] \geq k\} \quad & \mbox{if it exists,}\\
    N+1\quad & \mbox{otherwise.}
  \end{cases}
\end{eqnarray*}
\begin{lemma}[\cite{Beame200238}]\label{lem:pred_succ_data_structure}
  We can construct, in $O(d \sqrt{\log d / \log\log d})$ time,
  an $O(d)$-space data structure 
	for an array $B$ of $d$ positive integer in $[1, N]$ in increasing order
	which can answer $\Pred{B}{k}$ and $\Succ{B}{k}$ in $O(\sqrt{\log d / \log\log d})$ time for any query $k \in [1, N]$.
\end{lemma}

\subsection{Query algorithm}\label{query_algorithm}

Our algorithm simulates the query algorithm for a plain string~\cite{SUPS}
with $O(m)$-space data structures.
We summarize our algorithm below.

Let $[s, t]$ be a query interval such that $1 \leq s \leq t \leq n$.
Firstly, we compute the number of MUPSs contained in $[s, t]$.
This operation can be done in $O(\sqrt{\log m / \log\log m})$ 
by using $\Succ{\mupsbeg}{s}$ and $\Pred{\mupsend}{t}$.

Let $\num$ be the number of MUPSs contained in $[s, t]$.
If $\num \geq 2$, then there is no SUPS for this interval (Corollary~{1} of \cite{SUPS}).
Suppose that $\num = 1$.
Let $S[i..j]$ be the MUPS contained in $[s, t]$.
If $S[i-z..j+z]$ is a palindromic substring, 
then $S[i-z..j+z]$ is the only $\sups$ for $[s, t]$ where $z = \max \{ i-s, t-j \}$.
Otherwise, there is no $\sups$ for $[s, t]$ (Lemma~{6} of \cite{SUPS}).
Since this candidate has a run as the center, 
we can check whether $S[i-z..j+z]$ is a palindromic substring or not in constant time 
after computing all run-centered maximal palindromes.
Suppose that $\num = 0$ (this case is based on Lemma~{7} of \cite{SUPS}).
Let $p = \Pred{\mupsend}{t}, q = \Succ{\mupsbeg}{s}$.
We can check whether each of $S[\mupsbeg[p]-t+\mupsend[p]..t]$ 
and $S[s..\mupsend[q]+\mupsbeg[q]-s]$ is a palindrome or not.
If so, the shorter one is a candidate of $\supss$.
Let $\ell$ be the length of the candidates.
Other candidates are the shortest $\mupss$ which contain the query interval $[s, t]$.
If the length of these candidates is less than or equal to $\ell$, 
we need to compute these candidates as $\supss$.
We can compute these $\mupss$ by using range minimum queries on $\mupslen[p+1, q-1]$.
Thus, we can compute all $\supss$ in linear time w.r.t. the number of outputs (see \cite{SUPS} in more detail).

We conclude with the main theorem of this paper.
\begin{theorem}
  Given $\rle{S}$ of size $m$ for a string $S$,
  we can compute a data structure of $O(m)$ space in
  $O(m(\log \rlesigma + \sqrt{\log m / \log\log m}))$ time
  so that subsequent SUPS queries can be answered
  in $O(\alpha+\sqrt{\log m / \log\log m})$ time,
  where $\rlesigma$ denotes the number of distinct
  RLE-characters in $\rle{S}$ and $\alpha$ the number of SUPSs to report.
\end{theorem}

\section{Faster algorithms for a variant of the SUPS problem}

In this section, we consider a variant of Problem~\ref{prob:SUPS_RLE}
where each query interval $[s, t]$ is given as a tuple representing
the left-end run $s_r$ that contains $s$,
the local position $s_p$ in the left-end run $s_r$ corresponding to $s$,
the right-end run $t_r$ that contains $t$,
and the local position $t_p$ in the right-end run $t_r$ corresponding to $t$.
Intuitively, queries are also run-length encoded in this variant of the problem.
Formally, this variant of the problem is defined as follows:
\begin{problem}[$\sups$ problem on run-length encoded strings with run-length encoded queries]
\label{prob:SUPS_RLE_faster}
\leavevmode %\par
\begin{description}
\item[Preprocess:] $\rle{S} = (a_1, e_1), \ldots, (a_m, e_m)$ of size $m$ representing a string $S$ of length $n$.
\item[Query:] Tuple $(s_r, s_p, t_r, t_p)$ representing
  query interval $[s, t]$~ $(1 \leq s \leq t \leq n)$,
  such that
	$s = \rlebp{s_r} + s_p - 1$,
	$t = \rlebp{t_r} + t_p - 1$,
	$1 \leq s_r \leq t_r \leq m$,
	$1 \leq s_p \leq e_{s_r}$, and
	$1 \leq t_p \leq e_{t_r}$.
\item[Return:] All SUPSs for interval $[s, t]$.
\end{description}
\end{problem}

In this section, we present two alternative algorithms for Problem~\ref{prob:SUPS_RLE_faster}.

\subsection{Further combinatorial properties on MUPSs}

The key to our algorithms is a combinatorial property of
maximal palindromes in RLE strings.
To show this property, we utilize the following result.

\begin{lemma}[\cite{ApostolicoBG95,MATSUBARA2009900}]
  \label{lem:suf_pal_arithmetic_progression}
	Let $S$ be any string of length $n$.
        The set of suffix palindromes of $S$
	can be partitioned into $O(\log n)$ disjoint groups,
        such that the suffix palindromes in the same group
        has the same shortest period.
        Namely, each group can be represented by a single arithmetic progression
	$\langle s, d, t \rangle$, 
        such that $s$ is the length of the shortest suffix palindrome
        in the group, $t$ is the number of suffix palindromes in the group,
        and $d$ is the common difference (i.e. the shortest period).
\end{lemma}
When there is only one element in a group (i.e. $t = 1$),
then we set $d = 0$.

By applying Lemma~\ref{lem:suf_pal_arithmetic_progression}
to $\rle{S} = (a_1, e_1) \cdots (a_m, e_m)$ where 
each RLE factor $(a_i, e_i)$ is regarded as a single character,
we immediately obtain the following corollary:

\begin{corollary} \label{coro:rle_suf_pal_arithmetic_progression}
  Let $S$ be any string and let $m$ be the size of $\rle{S}$.
  The set of RLE-bounded suffix palindromes of $\rle{S} = (a_1, e_1) \cdots (a_m, e_m)$
  can be partitioned into $O(\log m)$ disjoint groups,
  such that the RLE-bounded suffix palindromes in the same group
  has the same shortest period.
  Namely, each group can be represented by a single arithmetic progression
  $\langle s', d', t' \rangle$, 
  such that $s'$ is the number of runs in the shortest RLE-bounded suffix palindrome
  in the group, $t'$ is the number of RLE-bounded suffix palindromes in the group,
  and $d'$ is the common difference (i.e. the shortest period).
\end{corollary}

In the sequel,
we will store arithmetic progressions representing
the RLE-bounded suffix palindromes of \emph{all}
RLE-bounded prefixes $\rle{S}[1..i] = (a_1, e_1),$ $\ldots, (a_i, e_i)$ of $\rle{S}$.
It may seem that it takes $O(m \log m)$ total space
due to Corollary~\ref{coro:rle_suf_pal_arithmetic_progression}.
However, since each RLE-bounded suffix palindrome of $\rle{S}[1..i]$
is a \emph{maximal} palindrome of string
$\rle{S} = (a_1, e_1), \ldots,  (a_m, e_m)$ where each run $(a_i, e_i)$ is
regarded as a single character,
and since there are only $2m-1$ such maximal palindromes in
$\rle{S} = (a_1, e_1), \ldots,  (a_m, e_m)$,
the total space requirement for storing all arithmetic progressions
is $O(m)$.

The next lemma is a key to our algorithms.

\begin{lemma} \label{lem:number_of_mups_in_a_run}
  Let $S$ be any string and let $\rle{S} = a_1^{e_1} \cdots a_m^{e_m}$.
  For any $1 \leq i \leq m$,
  the number of MUPSs of $S$ that end in the $i$th run $a_i^{e_i}$ of $\rle{S}$,
  namely in the position interval $[\rlebp{i}, \rleep{i}]$ in $S$, is $O(\log m)$.
\end{lemma}

\begin{proof}
  It is possible that the $i$th run $a_i^{e_i}$ itself is a MUPS.
  To consider other MUPSs ending in the $i$th run,
  we consider MUPSs of $S$ that begin \emph{before} the $i$th run $a_i^{e_i}$
  and end in $a_i^{e_i}$,
  namely, those MUPSs that begin in position range $[1..\rleep{i-1}]$
  and end in $[\rlebp{i}..\rleep{i}]$.
  We observe that
  any \emph{palindromic substring} in $S$ that begins before the $i$th run $a_i^{e_i}$
  and ends in the $i$th run $a_i^{e_i}$
  can be obtained by extending some \emph{suffix palindrome} of
  $S[1..\rleep{i-1}]$
  to the left and to the right within $S$.
  We remark that the extension may not terminate at RLE boundaries
  and can terminate within runs.

  For any $1 < i \leq m$,
  let $\RBSP_{i-1}$ be the set of RLE-bounded suffix palindromes
  of $\rle{S}[1..i-1] = (a_1, e_1), \ldots, (a_{i-1}, e_{i-1})$
  whose beginning positions coincide with the beginning positions
  of some runs in $\rle{S}[1..i-1]$.
  Now it follows from Corollary~\ref{coro:rle_suf_pal_arithmetic_progression} that
  the lengths of the suffix palindromes in $\RBSP_{i-1}$ can be represented by $O(\log m)$
  arithmetic progressions.
  
  Let $\langle s', d', t' \rangle$ be a single arithmetic progression
  representing a group of suffix palindromes in $\RBSP_{i-1}$.
  In what follows, we show that for each $\langle s', d', t' \rangle$,
  the number of MUPSs ending in the $i$th run $a_i^{e_i}$
  that can be obtained by extending elements of $\langle s', d', t' \rangle$ is at most two.
  The case where $t' \leq 2$ is trivial,
  and hence let us consider the case where $t' \geq 3$.
  We consider the following sub-cases:

  \begin{enumerate}
		\item When $a_{i - (s' + (t' - 1)d') - 1} = a_{i}$, 
	          then we can obtain a palindrome that ends in the $i$th run $a_i^{e_i}$
		  by extending the longest palindrome $P$
                  belonging to $\langle s', d', t' \rangle$.
                  This extended palindrome can be a MUPS of $S$ that ends in the $i$th run.
                  We note that there is a unique positive integer $\ell$
                  such that $a_i^\ell P a_i^\ell$ is a MUPS in $S$.                  

		\item When $a_{i - s' - 1} = a_{i}$,
                  then for any $2 \leq j \leq t'$, we concider the palindrome that ends in
                  the $i$th run $a_i^{e_i}$ by maximally extending the $j$th longest palindrome
                  $S[\rlebp{i - (s' + (t' - j)d')}..\rleep{i - 1}]$ belonging  
                  to $\langle s', d', t' \rangle$.
                  The length of this extension is $\min\{e_{i - s' - 1}, e_i\}$
                  to either side.
                  Now, for any $3 \leq k\leq t'$, we have that 
	          \begin{align*}
                    & S[\rlebp{i \! - \! (s' \! + \! (t' \! - \! k)d')} \! - \! \min\{e_{i - s' - 1}, e_i\}
				..\rleep{i \! - \! 1} + \min\{e_{i - s' - 1}, e_i\}] = \\
		    & S[\rlebp{i \! - \! (s' \! + \! (t' \! - \! (k-1))d')} \! - \! \min\{e_{i - s' - 1}, e_i\}
		      ..\rleep{i \! - \! d' \! - \! 1} \! + \! \min\{e_{i - s' - 1}, e_i\}].
                  \end{align*}
                  This implies that for any $3 \leq k\leq t'$
                  any palindrome that is obtained by extending
                  the $k$th longest palindrome corresponding to $\langle s', d', t' \rangle$
                  and ending in the $i$th run $a_i^{e_i}$,
                  occurs at least twice in $S$.
                  Thus, the elements of $\langle s', d', t' \rangle$ except for
                  the longest one and the second longest one do not yield MUPSs
                  ending in the $i$th run $a_i^{e_i}$.
	\end{enumerate}
        Consequently, for each single group $\langle s', d', t' \rangle$,
        there are at most two MUPSs ending in the $i$th run $a_i^{e_i}$
        that can be obtained by extending palindromes belonging to $\langle s', d', t' \rangle$.
        It follows from Corollary~\ref{coro:rle_suf_pal_arithmetic_progression}
        that there are $O(\log i)$ groups for each $i$.
        Thus, the number of MUPSs that end in the $i$th run is bounded by $O(\log m)$
        for any $1 \leq i \leq m$.
\end{proof}

See Fig.~\ref{fig:lemma11-2}, \ref{fig:lemma11-3}, and \ref{fig:lemma11-4}
for concrete examples for the proof of Lemma~\ref{lem:number_of_mups_in_a_run}.

\begin{figure}[!ht]
	\centerline{\includegraphics[width=0.9\textwidth]{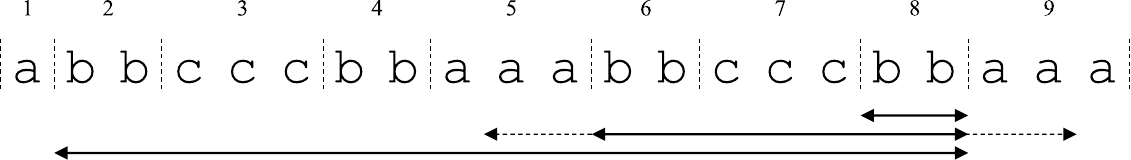}}
	\caption{
	  An example for the first case ($a_{i - (s' + (t' - 1)d') - 1} = a_{i}$) in the proof for 
          Lemma~\ref{lem:number_of_mups_in_a_run}.
	  Let $\rle{S} = \mathtt{a^1b^2c^3b^2c^3b^2c^3b^2a^3}$.
	  Let $i = 9$ and consider MUPSs that end in the $9$th run in $\rle{S}$.
          For $i-1 = 8$, the RLE-bounded suffix palindromes of $\rle{S}[1..8]$ are
          depicted by the solid arrows.
          Among those, we consider $\mathtt{b^2c^3b^2}$ that forms
          a single arithmetic progression $\langle s', d', t' \rangle = \langle 3, 0, 1 \rangle$.
          Since $a_{i - (s' + (t' - 1)d') - 1} = a_{5} = a_{9} = a_{i} = \mathtt{a}$,
          we can obtain palindrome $\mathtt{a^2b^2c^3b^2a^2}$
          ending in the $9$th run by extending
          the longest element belonging to this arithmetic progression (the extension is depicted by the broken arrows).
          This extended palindrome is a MUPS of $S$ ending in the $9$th run
          (namely, $\ell = 2$ is the unique extension length for which
          $a^{\ell}b^2c^3b^2a^{\ell}$ becomes a MUPS).
	}
\label{fig:lemma11-2}
\end{figure}

\begin{figure}[!ht]
	\centerline{\includegraphics[width=0.9\textwidth]{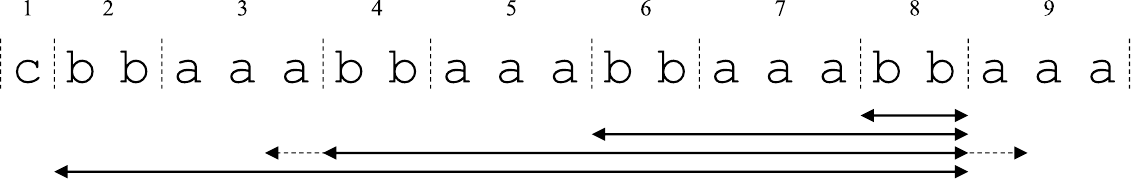}}
	\caption{
	  An example for the second case ($a_{i - s' - 1} = a_{i}$) in the proof for 
          Lemma~\ref{lem:number_of_mups_in_a_run}.
	  Let $\rle{S} = \mathtt{c^1b^2a^3b^2a^3b^2a^3b^2a^3}$
	  Let $i = 9$ and consider MUPSs that end in the $9$th run in $\rle{S}$.
          For $i-1 = 8$,
          consider the RLE-bounded suffix palindromes of $\rle{S}[1..8]$
          that can be represented by an arithmetic progression $\langle s', d', t' \rangle = \langle 1, 2, 4\rangle$.
          These RLE-bounded suffix palindromes are depicted by the solid arrows.
          Since $a_{i - s' - 1} = a_{7} = a_{9} = a_{i}  = \mathtt{a}$,
          we can obtain a palindrome ending in the $9$th run by extending
          the second longest element belonging to this arithmetic progression (the extension is depicted by the broken arrows).
          This extended palindrome is a MUPS of $S$ ending in the $9$th run.
          Since the extensions of the third and fourth longest elements 
          are both prefixes and suffixes of the above MUPS,
          no MUPSs can be obtained from the third and fourth longest elements belonging to this arithmetic progression.
          Also, since $a_{i - (s' + (t' - 1)d') - 1} = a_{1} = \mathtt{c} \neq  \mathtt{a} = a_{9} = a_{i}$ in this string,
          the longest element belonging to this arithmetic progression cannot be extended.
	}
\label{fig:lemma11-3}
\end{figure}

\begin{figure}[!ht]
  \centerline{\includegraphics[width=0.9\textwidth]{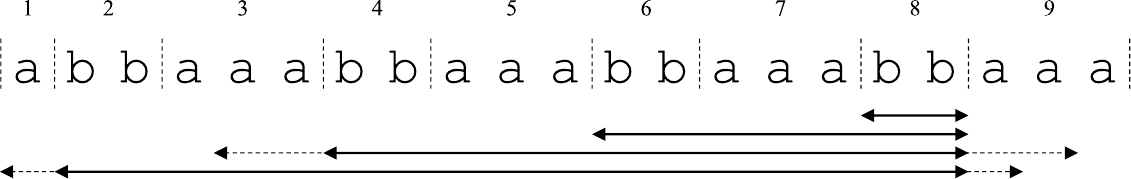}}
  \caption{
          An example where
          both the first case ($a_{i - (s' + (t' - 1)d') - 1} = a_{i}$) and
          the second case ($a_{i - s' - 1} = a_{i}$) in the proof for 
          Lemma~\ref{lem:number_of_mups_in_a_run} apply.
          Let $\rle{S} = \mathtt{a^1b^2a^3b^2a^3b^2a^3b^2a^3}$.
          Let $i = 9$ and consider MUPSs that end in the $9$th run in $\rle{S}$.
          For $i-1 = 8$,
          consider the RLE-bounded suffix palindromes of $\rle{S}[1..8]$
          that can be represented by an arithmetic progression $\langle s', d', t' \rangle = \langle 1, 2, 4\rangle$.
          These RLE-bounded suffix palindromes are depicted by the solid arrows.
          In this case,
          the extensions from the longest one and the second longest one
          belonging to the arithmetic progression are MUPSs
          of $S$ ending in the $9$th run (the extensions are depicted by the broken arrows).
          Note that the extensions from the third longest one and the fourth longest one
          occur at least twice in $S$.
	}
  \label{fig:lemma11-4}
\end{figure}

\subsection{Query algorithm}

Here we present our algorithms for answering queries of Problem~\ref{prob:SUPS_RLE_faster}.
While we inherit the basic concepts of our query algorithm for the original problem (Problem~\ref{prob:SUPS_RLE}),
we here use slightly different data structures.

\subsubsection{Preprocessing}

As was done in Section~\ref{sec:query_algorithm},
we store the set $\M_S$ of MUPSs of $S$ in the three following arrays.
We store $\M_S$ using the three following arrays:
\begin{itemize}
	\item $\mupsbeg[i]$ : the beginning position of the $i$th MUPS in $\M_S$.
	\item $\mupsend[i]$ : the ending position the $i$th MUPS in $\M_S$.
	\item $\mupslen[i]$ : the length of the $i$th MUPS in $\M_S$.
\end{itemize}
Since the number of MUPSs in $\M_S$ is at most $m$ (Corollary~\ref{coro:num-of-mupss}),
the length of each array is at most $m$.

Additionally, we build the two following arrays of size exactly $m$ each,
such that for each $1 \leq i \leq m$:
\begin{itemize}
              \item $\mrb[i]$ stores
                a sorted list of the beginning positions of MUPSs that begin in the $i$th run
                (i.e. in the position interval $[\rlebp{i}..\rleep{i}]$),
                arranged in increasing order.

	      \item $\mre[i]$ stores
                a sorted list of the ending positions of MUPSs that end in the $i$th run
                (i.e. in the position interval $[\rlebp{i}..\rleep{i}]$),
                arranged in increasing order.
\end{itemize}
We can easily precompute $\mrb$ (resp. $\mre$) in $O(m)$ time
by a simple scan over $\mupsbeg$ (resp. $\mupsend$).

Given a query input $(s_r, s_p, t_r, t_p)$ for Problem~\ref{prob:SUPS_RLE_faster},
we can retrieve the corresponding query interval $[s, t]$ over the string $S$
by
\begin{itemize}
	\item $s = \rlebp{s_r} + s_p - 1$,
	\item $t = \rlebp{t_r} + t_p - 1$.
\end{itemize}
Thus, provided that $\rlebp{i}$ are already computed for all $1 \leq i \leq m$,
we can retrieve the query interval $[s, t]$ in $O(1)$ time.
We can easily compute $\rlebp{i}$ for all $1 \leq i \leq m$
in $O(m)$ total time by scanning $\rle{S}$.

\subsubsection{SUPSs queries}

Suppose that we have retrieved the query interval $[s, t]$
from the query input $(s_r, s_p, t_r, t_p)$ as above.
The next task is to compute the number of MUPSs contained in $[s, t]$.

As was discussed previously, the number of MUPSs contained in $[s, t]$
can be computed from $\Succ{\mupsbeg}{s}$ and $\Pred{\mupsend}{t}$.
In our algorithm for Problem~\ref{prob:SUPS_RLE} (Section~\ref{sec:query_algorithm}),
we computed $\Succ{\mupsbeg}{s}$ and $\Pred{\mupsend}{t}$
using the successor/predecessor data structures of Lemma~\ref{lem:pred_succ_data_structure}
built on $\mupsbeg$ and $\mupsend$.
Here, we present two alternative approaches.

Our first solution is the following:
\begin{theorem}
	\label{theorem:problem2-1}
	Given $\rle{S}$ of size $m$ for a string $S$, 
	we can compute a data structure of $O(m)$ space in
	$O(m(\log \rlesigma + \sqrt{\log m / \log\log m}))$ time
	so that subsequent run-length encoded SUPS queries of Problem~\ref{prob:SUPS_RLE_faster}
        can be answered
	in $O(\sqrt{\log\log m / \log\log\log m} + \alpha)$ time,
	where $\rlesigma$ denotes the number of distinct
	RLE-characters in $\rle{S}$ and $\alpha$ the number of SUPSs to report.
\end{theorem}

\begin{proof}
  For each $1 \leq i \leq m$,
  we build the successor data structure of Lemma~\ref{lem:pred_succ_data_structure}
  on the elements stored in $\mrb[i]$.
  Similarly, 
  for each $1 \leq i \leq m$,
  we build the predecessor data structure of Lemma~\ref{lem:pred_succ_data_structure}
  on the elements stored in $\mre[i]$.
  Then, we compute $\Succ{\mupsbeg}{s}$ and $\Pred{\mupsend}{t}$
  from the successor/predecessor data structures for 
  $\Succ{\mrb[s_r]}{s}$ and $\Pred{\mre[t_r]}{t}$, respectively.
  This approach covers the cases where $\Succ{\mupsbeg}{s}$ exists in $\mrb[s_r]$
  and $\Pred{\mupsend}{t}$ exists in $\mre[t_r]$.
  To deal with the case where $\Succ{\mupsbeg}{s}$ does not exist in $\mrb[s_r]$,
  we precompute $\Succ{\mupsbeg}{\rleep{s_r}}$.
  We can precompute $\Succ{\mupsbeg}{\rleep{i}}$ for all $1 \leq i \leq m$
  in $O(m)$ time by a simple scan over $\mupsbeg$.
  The case where $\Pred{\mupsend}{t}$ does not exist in $\mre[t_r]$
  can be treated similarly.

  For each $1 \leq i \leq m$,
  let $c_i$ be the number of MUPSs stored in $\mrb[i]$.
  The successor data structure of Lemma~\ref{lem:pred_succ_data_structure}
  for the list of MUPSs in $\mrb[i]$ occupies $O(c_i)$ space,
  can be built in $O(c_i \sqrt{\log c_i / \log\log c_i})$ time,
  and answers successor queries in $O(\sqrt{\log c_i / \log\log c_i})$ time.
  It follows from Lemma~\ref{lem:number_of_mups_in_a_run} and Corollary~\ref{coro:num-of-mupss} that
  $c_i = O(\log m)$ for each $1 \leq i \leq m$
  and the total number of MUPSs stored in the data structures
  for all $1 \leq i \leq m$ is $\sum_{i=1}^{m}c_i = O(m)$,
  Therefore,
  the successor data structures for $\mrb[i]$ for all $1 \leq i \leq m$
  can be built in a total of $O(m\sqrt{\log m / \log\log m})$ time (due to Jensen's inequality),
  can be stored in $O(m)$ total space,
  and answer successor queries in $O(\sqrt{\log\log m /\log\log\log m})$ time.
  The same argument holds for the predecessor data structures for $\mre[i]$.
  Thus, we can count the number of MUPSs contained in the interval $[s,t]$
  in $O(\sqrt{\log\log m /\log\log\log m})$ time.
  The rest of our query algorithm is the same as in Section~\ref{query_algorithm}.
\end{proof}

Our second solution is simpler and can be built faster than the first solution,
but supports slightly slower queries.
\begin{theorem}
	Given $\rle{S}$ of size $m$ for a string $S$,
	we can compute a data structure of $O(m)$ space in
	$O(m\log \rlesigma)$ time
	so that subsequent run-length encoded SUPS queries of Problem~\ref{prob:SUPS_RLE_faster}
	in $O(\log\log m + \alpha)$ time,
	where $\rlesigma$ denotes the number of distinct
	RLE-characters in $\rle{S}$ and $\alpha$ the number of SUPSs to report.
\end{theorem}

\begin{proof}
  Given $s_r$ and $t_r$,
  we binary search $\mrb[s_r]$ and $\mre[t_r]$ for $\Succ{\mupsbeg}{s}$ and $\Pred{\mupsend}{t}$,
  respectively.
  Since the numbers of elements stored in $\mrb[s_r]$ and in $\mre[t_r]$ are $O(\log m)$ each
  by Lemma~\ref{lem:number_of_mups_in_a_run},
  the binary searches terminate in $O(\log \log m)$ time.
  The cases where $\Succ{\mupsbeg}{s}$ does not exist in $\mrb[s_r]$,
  and $\Pred{\mupsend}{t}$ does not exist in $\mre[t_r]$
  can be treated similarly as in Theorem~\ref{theorem:problem2-1}.
  The rest of our query algorithm follows our method in Section~\ref{query_algorithm}.
  Clearly, this data structure takes $O(m)$ total space.
\end{proof}

\section*{Acknowledgments}
This work was supported by JSPS KAKENHI Grant Numbers JP18K18002 (YN), JP17H01697 (SI), JP16H02783 (HB), JP18H04098 (MT), and by JST PRESTO Grant Number JPMJPR1922 (SI).

\bibliographystyle{abbrv}
\bibliography{ref}

\end{document}